\documentclass[a4paper]{article}

\usepackage[english]{babel}
\usepackage[utf8]{inputenc}
\usepackage{amsmath}
\usepackage{amssymb}
\usepackage{mathtools}
\usepackage{graphicx}
\usepackage[colorinlistoftodos]{todonotes}
\usepackage{cancel}
\usepackage{caption}
\usepackage{subcaption}

\newtheorem{theorem}{Theorem}

\newtheorem{proof}{Proof}

\newcommand{\bra}[1]{{\left\langle{#1}\right\vert}}
\newcommand{\ket}[1]{{\left\vert{#1}\right\rangle}}

\newcommand{\leftpar}{\boldsymbol{ \left(\right. }}
\newcommand{\rightpar}{\boldsymbol{ \left.\right) } }
\providecommand{\myfloor}[1]{\left \lfloor #1 \right \rfloor }

\begin{document}

\title{Fredkin Spin Chain}
\author{Olof Salberger and Vladimir Korepin}
\maketitle

\begin{abstract}
We introduce a new model of interacting spin 1/2. It describes interaction of three nearest neighbors. The Hamiltonian can be expressed in terms of Fredkin gates. The Fredkin gate (also known as the CSWAP gate) is a computational circuit suitable for reversible computing. Our construction  generalizes the work   ~\cite{movassagh2015power}.
Our model can be solved by means of Catalan combinatorics in the form of random walks on upper half of a square lattice [Dyck walks]. Each Dyck path can be mapped on a wave function of spins.   The ground state is an equally weighted superposition of Dyck walks [instead of Motzkin walks]. We can also express it as a matrix product state.   We further construct the model  of interacting spins 3/2 and greater half-integer spins. The models with higher spins require coloring of Dyck walks. We construct SU(k) symmetric model [here k is the number of colors]. The leading term of the entanglement entropy is then proportional to the square root of the length of the lattice [like in Shor-Movassagh model]. The gap closes as a high power  of the length of the lattice.
\end{abstract}

\newpage

\section{Introduction}

Entanglement is a feature  separating classical and quantum worlds. It was discovered 80 years ago, but it  becomes important now because of building quantum devices.
Entanglement is used in quantum communication and computation. 
Entropy of a subsystem is a good measure of entanglement ~\cite{bennett1996concentrating}, it is called entanglement entropy. It was studied in models of interacting spins. Quantum fluctuations in spin models enhance with lowering of space dimension. Spin chains provide good models for the 
experimental realization of quantum devices such as quantum computers ~\cite{schauss2015crystallization} ~\cite{hild2014far} ~\cite{nascimbene2012experimental} ~\cite{porras2004effective}. Entanglement entropy in spin chains was studied intensely.  Recently Peter Shor and Ramis Movassagh constructed an spin chain describing local  interaction of integer spins  ~\cite{movassagh2015power}, ~\cite{bravyi2012criticality}. It display abnormally high levels of entanglement.  We generalize their results to half-integer spins.

\section{The Hamiltonian and ground state}

Our bulk Hamiltonian is given by an interaction between singlet pairs with neighbouring sites. Specifically, the Hamiltonian can be written as 

\begin{equation}
H = H_{bulk} + H_{\partial} = H_{\partial} + \sum_j H_j  
\end{equation}

where we have the boundary term $H_{\partial} = \ket{\downarrow_1}\bra{\downarrow_1} +\ket{\uparrow_N}\bra{\uparrow_N}$, and where the bulk term is given by the following three-site interaction between singlet pairs and neighbouring sites:

\begin{equation}
H_j = \ket{\uparrow_j}\bra{\uparrow_j} \otimes \ket{S_{j+1,j+2}}  \bra{S_{j+1,j+2}} +  \ket{S_{j,j+1}} \bra{S_{j,j+1}} \otimes \ket{\downarrow_{j+2}}\bra{\downarrow_{j+2} }  
\end{equation}

where $\ket{S_{i,j}}$ is the familiar singlet state $\frac{1}{\sqrt{2}}(\ket{\uparrow_i}\ket{\downarrow_j} - \ket{\downarrow_i}\ket{\uparrow_j})$. In other words, we have an interaction term for a spin up on the left of a singlet, and a term for a spin down on the right of a singlet.\\

Despite its simplicity, this seemingly innocent Hamiltonian turns out to have a highly entangled ground state with a very rich combinatorial structure, as we will see in the following section which introduces the necessary formalism to write down the ground state.\\

Before we move on however, we allow ourselves to rewrite the bulk Hamiltonian in a few alternate ways. The first is in terms of Pauli spin operators, where it can be written as:

\begin{equation}
H_j = (1 + \sigma^Z_j)(1 - \vec{\sigma}_{j+1} \cdot \vec{\sigma}_{j+2}) + (1 - \vec{\sigma}_{j} \cdot \vec{\sigma}_{j+1})(1 - \sigma^Z_{j+2})  
\end{equation}

finally, we have yet another way to write down the bulk terms of the Hamiltonian, by noting their relationship to Fredkin (Controlled-Swap) gates:

\begin{equation}
H_j = (1 - F_{j,j+1,j+2}) + (1 -\sigma^x_{j+2}  F_{j+2,j+1,j} \sigma^x_{j+2})
\end{equation}

Where $F_{i,j,k}$ is the Fredkin gate acting on three qbits, which is both Hermitian and unitary. It swaps j,k if site i is in state $\ket{\uparrow}$ and does nothing if it is in state $\ket{\downarrow}$.

\subsection{Combinatorial prerequisites: Dyck Words}

In this section we consisely describe the of combinatorics necessary to describe the ground state and its relation to the Hamiltonian.

The central idea for defining our model is to think of the basis states of our length N chain of two-state systems as the set of N-symbol words in a two-letter alphabet. Depending on context, the letters would then be called {\bf } $\uparrow$ and $\downarrow$ for the basis states $\ket{\uparrow}$ and $\ket{\downarrow}$, 0 and 1 for $\ket{0}$ and $\ket{1}$, or more suggestively for our purposes, ( and ) for $\ket{ \leftpar }$ and $\ket{ \rightpar }$. \\

Among generic two-letter words, we may single out the set of Dyck words. \textit{These are precisely the set of words consisting of an equal number of left and right parentheses, such that every left parenthesis  has a matching right parenthesis further right along the word.} That is, there are no right or left parentheses that are  unmatched on either side of the word.\\

For example, ()(()) and (()()) are both valid six-letter Dyck words, while )(()() or ())(() are not.\\

Clearly, a word can only be Dyck if N = 2n is even, however since splitting up a word into pieces that may be odd is a useful operation we will still consider words of odd length.\\

The Dyck-ness of a word is thus a global feature which forces symbols at different places in the word to be correlated. For example, if we slice up the word along the middle the number of unmatched left parentheses in the left half must be equal to the number of unmatched right parentheses in the right half. These correlations will be the key source of entanglement in the physical model that we propose below.

\subsubsection{Dyck words as an equivalence class}
\label{sec:eqclass}

Having defined Dyck words and observed that their definition is global, we then observe that they do admit a local description.

To do this, we consider the larger set of arbitrary words and introduce the following moves, which in analogy to Reidemeister moves we call the Fredkin moves:

\begin{eqnarray} \label{FredkinMoves}
()) \leftrightarrow )() \\ (() \leftrightarrow ()(
\end{eqnarray}

where we allow a matched pair of adjacent parentheses to be moved anywhere. The reason why we call these the Fredkin moves is that when viewed as bit strings, the involution that takes one side to the other is a three-bit reversible logic gate known as the Fredkin gate, up to a flipped control bit. We will sometimes also refer to these as Fredkin relations.\\

We now consider how these relations divide up the set of words of length N into equivalence classes. We immediately observe that the relations preserve the matchings of the parentheses, and as such map Dyck words to Dyck words only and non-Dyck words to non-Dyck words only.\\

A full classification of the equivalence classes may be performed fairly straightforwardly by observing that:\\

1) If any given word has at least one pair of adjacent matched parentheses, we may use the Fredkin relation to move it to the front of the word. This allows us to focus on a subword of length N-2\\

2) The only words that have no matched pair of adjacent parentheses are of the form , i.e. consist of a left parentheses followed by b right parentheses.\\

Together, these observations imply that every word of length N is equivalent to exactly one word consisting of 2k matched parentheses, followed by $a$ unmatched right parenthese, followed by b left parentheses. That is, of the form ()()...k times () ()  )))... a times ...)) ((((... b times...(((, where $a + b + 2k = N$. We call this the \textit{standard form} of the word. Since clearly none of these words are equivalent to each other this completes our classification. We may denote these equivalence classes $C_{a,b}(N)$, and observe that $C_{0,0}(N)$, which is non-empty for even N, is exactly the set of Dyck words of length N = 2n.\\

We finish off by stating that the number of words in each equivalence class $C_{a,b}(N)$ is given by:
\begin{equation} \label{PathCount}
|C_{a,b}(N)| = \begin{cases} 
      \binom{N}{\frac{N + a + b}{2}} - \binom{N}{\frac{N + a + b}{2} + 1} & \textrm{ if $N - a - b$ is even and positive} \\
      1 & \textrm{if N = a + b} \\
      0 & \textrm{ otherwise} \\
   \end{cases}
\end{equation}

which in the nonzero cases match the numbers in the Catalan triangle (where each number is the sum of the numbers above it and the one to its left). We note that in the special case where N = 2n and a = b = 0, we recover the famous Catalan numbers $C_N = \frac{1}{n + 1}\binom{2n}{n}$.\\

We may prove this by constructing an explicit bijection between $C_{a,b}$ and $C_{0,a+b}$, which incidentally is also a symmetry of the bulk Hamiltonian. This simply consists of exchanging each unmatched right parenthesis in a member of $C_{a,b}$  with an unmatched left parenthesis. Since the subword between one unmatched parenthesis of a word and the next must be a Dyck word, this must map matched parentheses to matched parentheses and unmatched parentheses to unmatched parentheses, which makes the inverse of the map trivial to find.\\

The number of elements in $C_{0,a+b}$ can then be obtained directly from Bertrand's ballot theorem with ties allowed \cite{wikiBertrand}, or by using the combinatorial definition of the numbers in Catalan's triangle.

\subsubsection{Dyck words reinterpreted as Dyck paths}

In this section, we introduce the path notation to describe Dyck words. This notation is particularly useful in handwritten work and has the advantage of being easier to parse. The path formalism is also generally the most convenient to use for counting the number of elements in a given Fredkin equivalence class. \\

The key observation is that two-letter words can also be interpreted as paths on the upper half plane of a lattice. Specifically, all words in the equivalence class $C_{a,b}(N)$ can be viewed as a path from the lattice point $(0,a)$ to the lattice point $(N,b)$ such that on each step the path can only travel diagonally from $(x,y)$ to $(x+1, y\pm 1)$, and such that it is always stays in the upper half plane $y \geq 0$. The bijection is given simply by associating left parentheses with a step up and right parentheses with a step down.\\

\begin{figure}
\includegraphics[width=\textwidth]{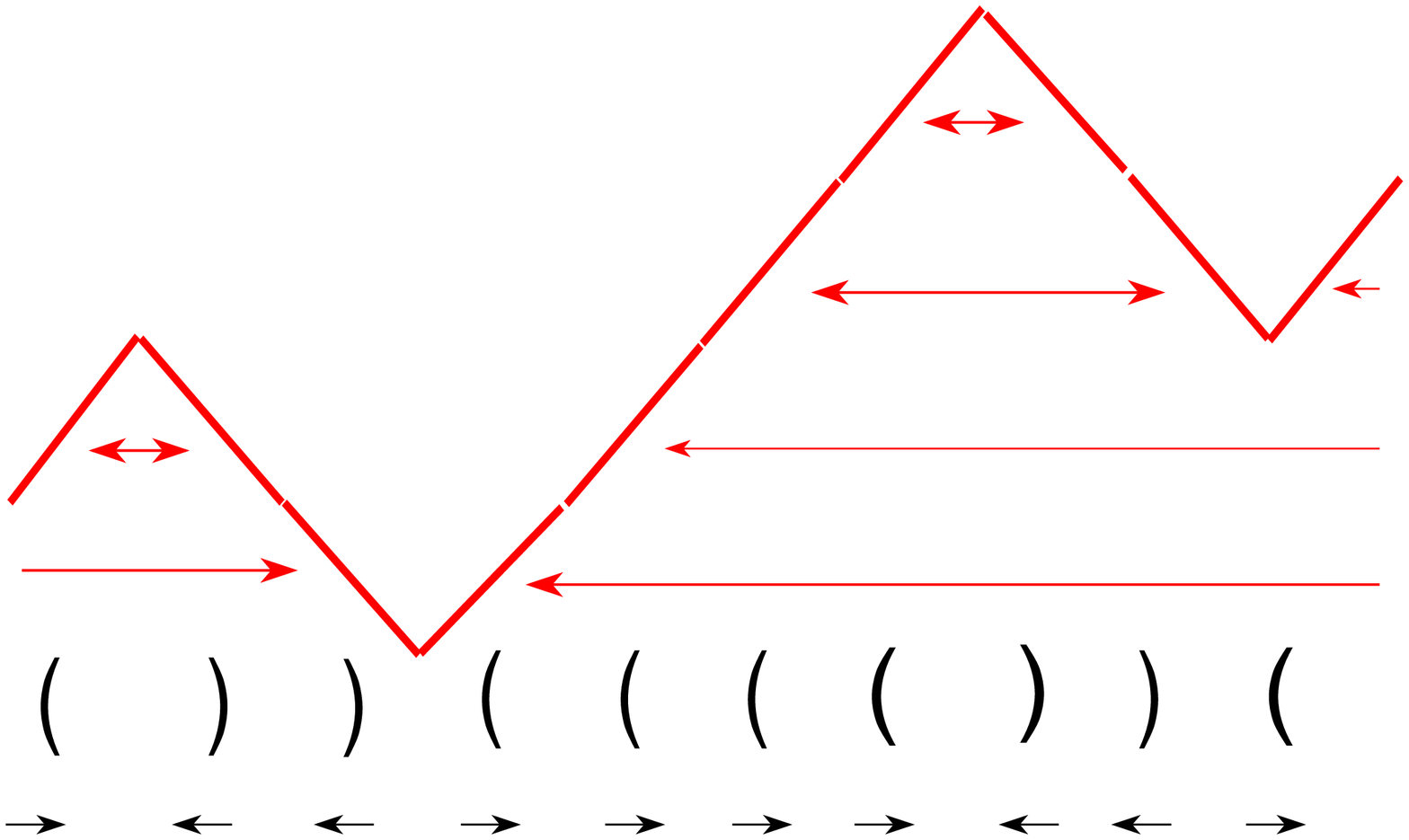}
\caption{\label{fig:Notations}In this figure we show the different ways in which our basis states can be expressed: as a sequence of spins up or down, as Dyck words, or as paths. It also shows the relation between matched parentheses and matched steps.}
\end{figure}

Parenthesis matching can be obtained in the following way: for any word of the form $W = w_1 ( w_2 ) w_3$, we have that $w_2$ must be a Dyck word if and only if the two parentheses we have written out are matched, since otherwise, either the left or right parentheses would instead be matched with an unmatched parenthesis of $w_2$. \\

As a result, the path corresponding to $w_2$ never goes further down than its starting point. The matching of parentheses can then be translated into matching steps in the following way: for a given up step from $y = y_0$ to $y = y_0 + 1$, the first down step after it which goes from $y = y_0 + 1$ to $y = y_0$, if it exists, is the \textit{matching down step} which corresponds to the matching right parenthesis. The two-dimensional nature of paths make matchings significantly easier to parse, at the expense of using a less compact notation.

\begin{figure}
\includegraphics[width=0.65\textwidth]{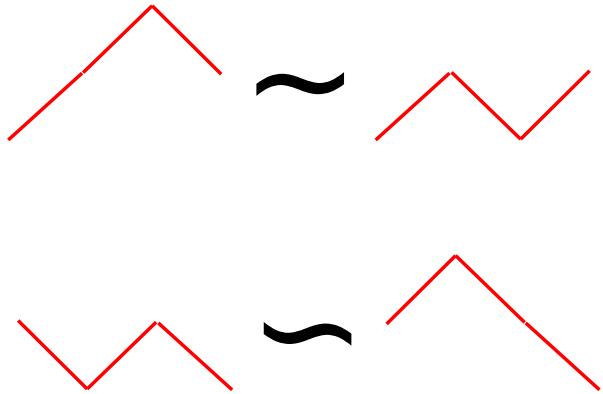}
\caption{\label{fig:FredkinMoves}The Fredkin moves in path notation.}
\end{figure}

\subsection{The ground state}

The ground state of our model with 2n sites $\ket{D_n}$ can be written as a uniform superposition of the basis states corresponding to Dyck words(or paths). 

To give some low-n examples, using the bracket notation where $\ket{\uparrow} = \ket{\leftpar}$ and $\ket{\downarrow} = \ket{\rightpar}$, we have:

\begin{align*}
\ket{D_{1}} = \ket{\leftpar \rightpar}\\
\ket{D_{2}} = \ket{\leftpar \rightpar\leftpar \rightpar} + \ket{\leftpar \leftpar\rightpar \rightpar}\\
\ket{D_3} = \ket{\leftpar \leftpar \leftpar \rightpar \rightpar \rightpar} + \ket{\leftpar \leftpar \rightpar \leftpar \rightpar \rightpar} + \ket{\leftpar\rightpar \leftpar \leftpar \rightpar  \rightpar} + \ket{\leftpar \leftpar  \rightpar \rightpar \leftpar\rightpar}+ \ket{\leftpar\rightpar \leftpar \rightpar\leftpar   \rightpar}
\end{align*}

and

\begin{align*}
\ket{D_4} = \ket{\leftpar \leftpar \leftpar\rightpar \leftpar  \rightpar \rightpar  \rightpar} + \\ 
\ket{\leftpar \leftpar \leftpar \rightpar \leftpar  \rightpar \rightpar  \rightpar} + \ket{\leftpar \leftpar \leftpar \rightpar \rightpar  \leftpar  \rightpar \rightpar} + \ket{\leftpar\leftpar \rightpar \leftpar  \leftpar  \rightpar \rightpar  \rightpar} +\\ \ket{\leftpar \rightpar\leftpar \leftpar  \leftpar  \rightpar \rightpar  \rightpar} + \ket{\leftpar \leftpar \leftpar \rightpar \rightpar  \rightpar \leftpar  \rightpar} + \ket{\leftpar \leftpar  \rightpar \rightpar \leftpar \leftpar  \rightpar  \rightpar} +\\
+ \ket{\leftpar \leftpar \rightpar \leftpar  \rightpar \leftpar \rightpar  \rightpar} 
+ \ket{\rightpar \leftpar\leftpar \leftpar   \rightpar \leftpar \rightpar  \rightpar} 
+ \ket{\leftpar \leftpar \rightpar  \leftpar \rightpar  \rightpar \leftpar  \rightpar} \\
+ \ket{\leftpar \leftpar  \rightpar  \rightpar  \leftpar \rightpar\leftpar \rightpar} 
+ \ket{\leftpar \rightpar \leftpar \leftpar  \rightpar  \rightpar\leftpar \rightpar} 
+ \ket{\leftpar \rightpar\leftpar \rightpar\leftpar \leftpar  \rightpar  \rightpar }  \\
+ \ket{\leftpar  \rightpar \leftpar  \rightpar \leftpar \rightpar \leftpar \rightpar} 
\end{align*}

which have 1, 2, 5, and 14 terms, respectively.

To see that it is the unique ground state, we do the following: first, we show that $\ket{D_n}$ is indeed a zero eigenstate of all individual terms in the Hamiltonian. Next, we observe that this implies frustration freeness: each term is positive semidefinite, therefore any other eigenstate with zero eigenvalues must be a zero eigenstate of all individual terms. Finally, we show that any state with that last property must be equal to $\ket{D_n}$. Alternatively, one can simply apply the orbit theorem in the appendix to the bulk Hamiltonian, and check which ground states are also ground states of the boundary terms.\\

We start with the boundary terms. Since every Dyck word must start with an opening parenthesis and end with a closing parenthesis(for walks, start with a step up and end with a step down), $\ket{D_n}$ is indeed a ground state of the boundary terms.\\

Next, we consider the bulk terms. Consider the term $$\ket{\uparrow_j}\bra{\uparrow_j} \otimes \ket{S_{j+1,j+2}}  \bra{S_{j+1,j+2}} = \big[ \ket{\leftpar \leftpar \rightpar}_j - \ket{ \leftpar \rightpar \leftpar}_j \big] \big[ \bra{\leftpar \leftpar \rightpar}_j - \bra{ \leftpar \rightpar \leftpar}_j \big] $$ , where the index j denotes the position of the first site. \\

This term is a projector onto the state $\ket{\leftpar \leftpar \rightpar}_j - \ket{ \leftpar \rightpar \leftpar}_j $. We then observe that we may generically write the Dyck state as $\ket{\psi_1} \otimes \ket{\leftpar \leftpar \rightpar}_j \otimes \ket{\psi_2} + \ket{\chi_1} \otimes \ket{\leftpar \rightpar \leftpar}_j \otimes \ket{\chi_2} + \ket{\psi} $, where $\ket{\psi}$ is orthogonal to both $\ket{\leftpar \leftpar \rightpar}_j$ and $\ket{ \leftpar \rightpar \leftpar}_j$. By definition, the last term is annihilated by our projector, so we may focus on the first two.

However, we note that the invariance of the set of Dyck words under the Fredkin moves \eqref{FredkinMoves} give us that if $w_1  w_2$  is a Dyck word, so is $w_1 \leftpar \rightpar \leftpar w_2$ and vice versa. Therefore, in our expression above we must have $\ket{\psi_1} = \ket{\chi_1}$ and $\ket{\psi_2} = \ket{\chi_2}$, and we have $\ket{D_n} = \ket{\psi_1} \otimes \big[ \ket{\leftpar \leftpar \rightpar}_j  + \ket{\leftpar \rightpar \leftpar}_j \big]\otimes \ket{\psi_2} + \ket{\psi}$. Both terms are orthogonal to  $\ket{\leftpar \leftpar \rightpar}_j - \ket{ \leftpar \rightpar \leftpar}_j $ and will be annihilated by our projector. The process is analogous for the second bulk term. \\

Therefore, the Dyck state is indeed a ground state of all individual terms of our Hamiltonian. Since such a state exists, any other ground state of the Hamiltonian must also be a ground state of all individual terms.\\

Now, we consider the possible ground states of the bulk terms. We now know that it must be a ground state of each individual term. We may then perform the above calculation in reverse: we break up an arbitrary ground state candidate $\ket{\psi}$ into $\ket{\Psi} = \ket{\psi_1} \otimes \ket{\leftpar \leftpar \rightpar}_j \otimes \ket{\psi_2} + \ket{\chi_1} \otimes \ket{\leftpar \rightpar \leftpar}_j \otimes \ket{\chi_2} + \ket{\psi} $, where $\ket{\psi}$ is orthogonal to both $\ket{\leftpar \leftpar \rightpar}_j$ and $\ket{ \leftpar \rightpar \leftpar}_j$. If this is annihilated by the projector $ \big[ \ket{\leftpar \leftpar \rightpar}_j - \ket{ \leftpar \rightpar \leftpar}_j \big] \big[ \bra{\leftpar \leftpar \rightpar}_j - \bra{ \leftpar \rightpar \leftpar}_j \big] $, we must have $\big[ \bra{\leftpar \leftpar \rightpar}_j - \bra{ \leftpar \rightpar \leftpar}_j \big] \ket{\Psi} = 0$. Expanding this, we get  $\ket{\psi_1} \otimes \bra{\leftpar \leftpar \rightpar}_j\ket{\leftpar \leftpar \rightpar}_j \otimes \ket{\psi_2} - \ket{\chi_1} \otimes \bra{\leftpar \rightpar \leftpar}_j \ket{\leftpar \rightpar \leftpar}_j \otimes \ket{\chi_2}$, so we come to the conclusion that for any such ground state $\ket{\chi_1} = \ket{\chi_2}$ and $\ket{\psi_1} = \ket{\psi_2}$. In other words, the parts of the ground state containing $\ket{\leftpar \rightpar \leftpar}_j$ and $\ket{\leftpar \leftpar \rightpar}_j $ are exactly equal. We can then simply repeat this for all bulk terms in the hamiltonian.\\

This is the key observation we needed to make. Writing out the ground state as $\ket{\Psi} = \sum_i c_i \ket{w_i}$, we have just proven that if the words $w_i$ and $w_j$ are related by a Fredkin move, then $c_i = c_j$. \\

To formulate this in terms of an explicit orthogonal basis for our ground states, we can use our discussion of equivalence classes under Fredkin moves in section \ref{sec:eqclass} to define the states $\ket{C_{a,b}(N)} = \sum_{w_i \in C_{a,b}(N)} \ket{w_i}$ which are direct superpositions of all states in the equivalence class $C_{a,b}(N)$. Then, we see that any ground state $\ket{\Psi}$ of the bulk terms can be written as  $\ket{\Psi} = \sum_{a,b} \ket{C_{a,b}(N)}$, so the states $\ket{C_{a,b}(N)}$ form a complete basis for the space of bulk term ground states.\\

Among these, for even N = 2n, only $\ket{C_{0,0}(N)} = \ket{D_n}$ which has no unpaired parentheses/steps is also a ground state of the boundary terms, which finally proves that the Dyck state is indeed the unique ground state of our Hamiltonian when the boundary terms are introduced.

\subsubsection{Schmidt decomposition and measures of entanglement}

The Dyck state is highly entangled. If we break up the chain into two halves with lengths L and N-L, where without loss of generalization we may assume that $L < N-L$. The Schmidt rank will then be $ \xi = \myfloor{\frac{L}{2}}$, while the the entanglement entropy will scale as $$S = \frac{1}{2} \log(L) + O(1)$$.\\

Our primary tool to study the entanglement between blocks will be the the Schmidt decomposition. The Dyck state can be decomposed into the states

\begin{equation}
\ket{D_n} = \sum_{m = 0 }^{ L} \sqrt{p_m} \ket{ C_{0,m}(L) } \otimes \ket{ C_{m,0}(N-L) }
\end{equation}

where we define the states $\ket{C_{a,b}(L)}$ as a direct generalization of the Dyck state, which is a direct superposition of the basis states corresponding to the equivalence classes  $C_{a,b}(N)$ that we defined in section \ref{sec:eqclass}.\\

In particular, we note that $C_{a,b}(L)$ is non-empty if and only if $a + b < L$ and the parity of $a + b$ matches that of L. As a result, the nonzero terms $\ket{ C_{0,m}(L) } \otimes \ket{ C_{m,0}(L) }$ will be those with even $m < L$ if $L$ is even, or odd $m < L$ if $L$ is odd. \\

This leads to our above expression for the Schmidt rank, since the number of even integers $\leq$ L if L is even is given by $\frac{L}{2}$, and the number of odd integers less than L if L is odd is given by $\frac{L-1}{2}$. The two expressions may be combined concisely with the floor function.\\

Now we may move on to compute the entanglement entropy. We start by deriving a useful expression for $p_m$, which must be the product of the normalization factors of $\ket{ C_{0,m}(L) }$ and $ \ket{ C_{m,0}(L) }$ divided by the normalization factor of $\ket{D_n}$. This gives us $p_m = \frac{|C_{0,m}(L)| |C_{0,m}(N-L)| }{|C_{0,0}(N)|}$.\\

To find a useful expression for the $p_m$, we will start by cleaning up our notation a bit. We have already written $N = 2n$. We will split the expression into two cases where $L$ is odd or even.\\

For the even case, we write $L = 2l$, and $m = 2h, h \in \left[ 0,l \right] $ as the $p_m$ will be zero for odd m.  Then for the $h < l$ case we may simplify:

\begin{align*}
|C_{0,m}(L)| = \binom{2l}{l + h} - \binom{2l}{l + h + 1} = \frac{2h + 1}{l + h + 1} \binom{2l}{l + h}  \\
|C_{m,0}(N-L)| = \binom{2(n - l)}{n - l+ h} - \binom{2(n - l)}{n - l + h + 1} =  \frac{2h + 1}{ n-l + h + 1} \binom{2(n-l)}{n - l + h}
\end{align*}

 while in the odd case we write $L = 2l + 1$ and $m = 2h + 1, h \in \left[ 0,l \right] $ which gives us
\begin{align*}
|C_{0,m}(L)| = \binom{2l+ 1}{l + h + 1} - \binom{2l + 1}{l + h + 2} = \frac{2h + 2}{l + h + 2} \binom{2l+1}{l + h + 1}  \\
|C_{m,0}(N-L)| = \binom{2(n-l) - 1}{(n-l) + h} - \binom{2(n-l) - 1}{(n-l) + h + 1} 
\end{align*}

which in the even case gives us:

\[
p_h = \frac{1}{Cat(n)} \frac{(2h + 1)^2}{(h + 1 + \frac{n}{2})^2 - (l -\frac{n}{2})^2} \binom{2l}{l + h} \binom{2(n-l)}{n - l + h}
\]

Our next step is then simply to use the Gaussian approximation for binomials $\binom{2n}{n+k} \approx \frac{4^n}{\sqrt[]{\pi n}} \exp(-\frac{k^2}{n})$, which gives us:

\[
p_h \approx \frac{4^n}{\pi Cat(n) \sqrt[]{l(n-l)}} \frac{(2h + 1)^2}{(h + 1 + \frac{n}{2})^2 - (l -\frac{n}{2})^2} \exp\left[-h^2 (\frac{1}{l} +\frac{1}{n-l}) \right]
\]

Which is a good approximation when n is large, and peaks when $\frac{h^2 n}{l(n-l)} \approx 1$. Since the Gaussian factor will suppress the parts where h is comparable to $\frac{n}{2}$, our final approximation becomes

\[
p_h \approx \frac{h^2}{Z} \exp\left[-h^2 \frac{n}{l(n-l)} \right]
\]

where $Z$ is a normalization factor.\\

This expression is largely analogous to the one in \cite{bravyi2012criticality},  \cite{movassagh2015power}, and from here on we may proceed in a manner fully analogous to the method outlined in these publications. The key idea is to approximate sums with integrals, and performing a change of variables to isolate the $L(N-L)$ dependence. In particular, we have:

\begin{align*}
S = -\sum_h p_h log(h) \approx -\int_0^{\infty} p_x log(p_x) dx = \left[ \sqrt{ \frac{l(n-l)}{l}} d\alpha = dx  \right] = \\
-\int_0^{\infty}  \sigma \rho_\alpha \log(\sigma \rho_\alpha)  \sigma^{-1} d\alpha = -\int_0^{\infty}   \rho_\alpha [\log(\sigma) +  \log(\rho_\alpha)] d\alpha = \\
- \log(\sigma) - \int_0^{\infty}   \rho_\alpha \log(\rho_\alpha) d\alpha  = \frac{1}{2} \log\left[ \frac{l(n-l)}{l} \right] -  \int_0^{\infty}   \rho_\alpha \log(\rho_\alpha) d\alpha
\end{align*}

Where $ p_x = \frac{x^2}{\tilde{Z}}\exp\left[-x^2 \frac{n}{l(n-l)} \right]$, $\tilde{Z}$ is a normalization factor so that $\int p_x dx = 1$ and $\alpha = \sigma x $ and $\sigma = \sqrt{\frac{N}{L(N-L)} } $. Furthermore, we define $\rho_{\alpha} = \frac{1}{Z'} \alpha^2 \exp[-\alpha^2]$ where $Z'$ is a constant such that $\int \rho_{\alpha} = 1$, which gives us $ p_x = \sigma \rho_{\alpha}$ (with a factor $\sigma^3$ from the normalization factor and a $\sigma^{-2}$ from $ x^2 = \sigma^{-2} \alpha^2 $ ), which we used to insert the correct sigma factors. Thus, when N and L are large enough for the integral approximation to be valid, we have:

\[
S = \frac{1}{2} \log \left[ \frac{L(N-L)}{N} \right] + O(1)
\]

\subsubsection{Matrix product state definition}

To find the MPS formulation we let ourselves be guided by the form of the Fredkin moves (\ref{FredkinMoves}). We would like the product $A^{(} A^{)}$ to commute with $A^{(}$ and $A^{)}$, or at least have a commutator which can easily be projected out, while having a more nontrivial expression for $A^{)} A^{(}$. Here $A$ is a local matrix in matrix state represettation of the ground state. It can be compared with $L$ operator of quantum inverse scattering method.\\

We can do this by setting $A^{(}_{ij} = \delta^{i+1,j}$ and $A^{)}_{ij} = \delta{i,j+1}$ for $ \infty > i,j \geq 1$ and 0 otherwise. In that case, before we truncate to a finite bond dimension, we have $A^{(} A^{)} = 1$ while $A^{)} A^{(}$ is a projector that does not commute with $A^{(}$ or $A^{)}$.\\

To obtain the Dyck state, we introduce the boundary (basis) vector $v_i = \delta^{i,1}$. Then we have that a the matrix element $v^{T} A^{k_1} A^{k_2}... A^{k_{N-1}} A^{k_N} v$ of the product will be nonzero if and only if the indices $k_1 k_2 ... k_{N-1} k_{N}$, which each are $\rightpar$ or $\leftpar$-valued, together form a Dyck word.\\

This is true because for any index values we can remove $A^{(} A^{)} = 1$ without changing the value of the product until we are left with an expression of the form $v^{T} {[A^{)}]}^a {[A^{(}]}^b v$. Hence the matrix product depends only on to which equivalence class of the word formed by the indices belongs to under the Fredkin moves, and using the boundary vector $v_i = \delta^{i,1}$ the relevant matrix element this will yield zero if the word belongs to any equivalence class other than that of Dyck words. Hence, the matrices above do indeed give a matrix product state description of the Dyck state.\\

We now only have to check whether we can truncate our infinite-dimensional matrices to finite-dimensional matrices to obtain an MPS description that is actually useful. This can be done by observing that the entries of $A_{ij}$ with i or j larger than N do not affect the product [ground state] $$v^{T} A^{i_1} A^{i_2}... A^{i_{N-1}} A^{i_N} v$$, meaning that we can simply truncate it to an N-dimensional MPS (which can be refined to N/2 exactly, and to $O(\sqrt{N})$ with arbitrarily small errors as N becomes large).



\subsection{Spin wave solutions}

The Fredkin model is closely related to the XXX Heisenberg model. In particular, since it commutes with z-component of total spin $Z = \sum_j \sigma_j^z$, we can straightforwardly apply the coordinate Bethe ansatz analyze the 1-magnon sector of the model. \\

In fact, since our model has two invariants $a$ and $b$, we can obtain a slightly stronger result, namely that all sectors with $a + b = N-2$ have a spectrum fully isomorphic to the 1-Magnon sector of a ferromagnetic Heisenberg XXX chain of length N-1.\\

The key insight is that when viewed as a path, these equivalence classes consist of a long dip with a single $/\backslash$  peak that may be moved around. The N-1 basis states of this class, which we call $\ket{\psi_{j}}$ , correspond to the path where the single peak has an up-step at site $j$ and the down-step at site $j+1$.\\

For each j, the only terms in the Fredkin Hamiltonian that can affect $\ket{\psi_{j}}$ non-trivially will be the three-site terms that act on both site j and j+1, which leads to four candidates. By brute force consideration of all cases, we will see that up to constants there will be exactly two terms remaining, one which maps $\ket{\psi_{j}}$ to $-\ket{\psi_{j+1}}$ and one which maps $\ket{\psi_j}$ to $-\ket{\psi_{j-1}}$. Therefore, the restriction of the Fredkin model Hamiltonian with open boundary conditions to this sector is isomorphic to the ferromagnetic Heisenberg XXX Hamiltonian with open boundary conditions restricted to the one-magnon sector, on a chain of length N-1.

\subsection{Ground states with alternative boundary conditions}

While we have been focusing on the ground states with boundary terms $H_{\partial} = \ket{\downarrow_1}\bra{\downarrow_1} +\ket{\uparrow_N}\bra{\uparrow_N}$, our model does exhibit some rather surprising properties with other boundary conditions.\\

\begin{figure}
\includegraphics[width=0.65\textwidth]{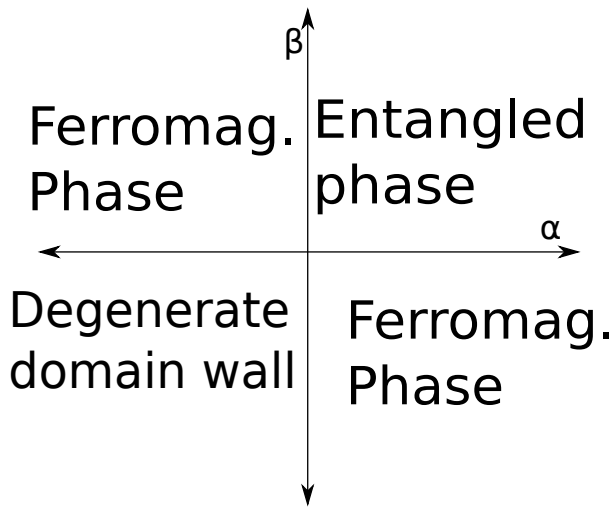}
\caption{\label{fig:phasediagram}A phase diagram describing the different phases of the model as we vary the boundary conditions.}
\end{figure}

The first generalization we could make is to a $H_{\partial}(\alpha, \beta) = (1 - \alpha\sigma^z_0) +(1 + \beta \sigma^z_N)$, which yields our previous boundary terms for $\alpha = \beta = 1$. For $\alpha = - \beta $, the bulk ground states $\ket{C_{N,0}}$ and $\ket{C_{0,N}}$ will be ground states of the case where $\alpha = 1$ and $\alpha = -1$, respectively. For $\alpha = \beta = -1$, we have an $N-2$ degeneracy as all the domain wall states $\ket{C_{a,b}}$ with $a + b = N$ and $a,b \geq 1$ are ground states. Since in call cases we have bulk ground states which are also boundary ground states, the model stays frustration free meaning these are the only possible ground states for these boundary conditions.

Furthermore, since the model is frustration free for all possible sign combinations of $\alpha, \beta$, the ground state depends \textit{only} on the signs of $\alpha, \beta$ and not on their magnitude. In other words, phase transitions may happen only when  $\alpha$ or $\beta$ is zero. These results may be combined into the phase diagram in figure \ref{fig:phasediagram}.\\

\subsubsection{Periodic boundary conditions}

We get a somewhat more surprising result with periodic boundary conditions. Here, we note that the Hamiltonian commutes with $Z = \sum_j \sigma^z_j$ and with the cyclic permutation operator $C^N$ which maps the site $j$ to the site $j+1$. We may use the eigenvalues of these two operators to index our ground states in the model with periodic boundary conditions.\\

For odd N, all ground states are eigenstates with eigenvalue 1 of $C^N$, and can be indexed solely using the $Z$ operator. In this case all ground states also happen to be eigenstates of the Heisenberg XXX model and are totally symmetric. The degeneracy in this case is N.\\

For even $N = 2n$, the situation is different, and we have a ground state degeneracy N+1. N of these are eigenvectors with unit eigenvalues of the $C^N$ operator and are also Heisenberg XXX ground states. However, there is one additional state, which we call the \textit{anomalous state} $\ket{D_{an}}$, which is an eigenstate of $C^N$ with eigenvalue -1, and a zero eigenstate of $Z$. Is is given by:

\[
\ket{D_{an}} = \frac{1}{\sqrt{\binom{2n}{n} } } \sum_{m=0}^n \sqrt{ |C_{m,m}| } (-1)^m  \ket{C_{m,m}}
\]

To see why this is the case, we consider what happens to the equivalence classes as we add Fredkin moves that continue periodically across the edge.\\

We may define a periodic matching across the edge, which is most easily done by going back to the Dyck word formalism: two parentheses are periodically matched if they are matched in the usual sense for at least one cyclic permutation of the Dyck word. Two steps are periodically matched if their corresponding parentheses in the Dyck word formalism are periodically matched.\\

The key idea is then that we can use our previous algorithm that reduces any path to a series of down steps (right parentheses), followed by a series of $/ \backslash $ peaks (matched pairs), followed by a series of up steps (left parentheses). However, whenever we have a down-step at site 1 and an up step at site N, we have a peak wound across the edge.\\

If we have any unmatched step adjacent to it, we can move this peak away from the edge to join the other peaks. By repeating this process, we can then always reduce to the case where all matched steps are adjacent to their match, and we are left with Z unmatched up steps and no unmatched down steps if Z is positive, or $|Z|$ unmatched down steps and no unmatched upsteps if Z is negative. \\

At this point, we are left with two cases: Either there is a peak wound across the edge, or there isn't. If $Z > 0$, we can move an unmatched up-step to site N-1, move the peak past it, and we will be left with no peaks wound across the edge. For $Z < 0$, we may perform the mirror image of this algorithm. However, for Z = 0 there is no way to transform one case into the other.\\

Thus, for $Z \neq 0$ we are left with exactly one equivalence class for each value of Z:

\[
C_Z = \begin{cases} 
      \cup_k C_{Z+k,k} & \textrm{ if $Z > 0$} \\
      \cup_k C_{k,Z+k} & \textrm{ if $Z < 0$} \\
   \end{cases}
\]

which all just contain all basis states that satisfy $\sum_j \sigma^z_j = Z$. Using the orbit theorem we see that these all generate the same ground states as the Heisenberg XXX Hamiltonian.

In the the Z = 0 case, we have two equivalence classes that classify whether the number of pairs matched across the edge is odd or even. Each of these gives us a ground state. The symmetric sum of these is the $Z = 0$ ground state of Heisenberg XXX, while the difference yields the anomalous state above.

\section{Generalizing to a colored SU(k)-symmetric model with a square root entanglement entropy leading term}
It is possible to generalize the Fredkin model to exhibit a colored Dyck state as its ground state in a manner largely analogous to ~\cite{movassagh2015power}, which yields an exponentially increasing Schmidt rank and entanglement entropy. However, since we consider Dyck paths rather than Motzkin paths the Hilbert space for each site will be of even dimension, which allows us to naturally write it as a tensor product of a two-dimensional Hilbert space and a k-dimensional "color" Hilbert space.\\

In particular our colored model has a manifest $U(1) \times SU(k)$ symmetry which generalizes the U(1) symmetry of the uncolored Fredkin model, where k is a positive integer giving the number of colors. \\

The expression for the Fredkin Hamiltonian above in terms of Pauli matrices can be generalized to an expression in terms of $SU(2) \otimes SU(c)$ matrices for each site j of the form $T^a_j = \sigma_j^z \otimes t_j^a$ where the $t^a$ are the generators of SU(k) in the fundamental representation. For convenience we will also use the $SU(2)$ matrices that by abuse of notation we write $ \sigma_j^i = \sigma_j^i \otimes 1_j$.\\

Finally, we define the projection operators $P^{\pm}_j = \frac{1 \pm \sigma^z_j}{2}$ and the cyclic permutation operator $C_{j_1,j_2, j_3}$ which cyclically permutes the sites $j_1, j_2, j_3$.\\

This allows us to conveniently write the Hamiltonian in a manifestly SU(k)-invariant way as:

\begin{equation}
H = H_F + H_X + H_{\partial}
\end{equation}

where

\begin{align*}
H_F = \sum_{j=1}^{N-2} P_{j}^{+}P_{j+1}^{+}P_{j+2}^{-} + P_{j}^{+}P_{j+1}^{-}P_{j+2}^{+} -P_{j}^{+}P_{j+1}^{+}P_{j+2}^{-}C_{j,j+1,j+2} - C^{\dagger}_{j,j+1,j+2}P_{j}^{+}P_{j+1}^{+}P_{j+2}^{-} + \\
+ P_{j}^{+}P_{j+1}^{-}P_{j+2}^{-} + P_{j}^{-}P_{j+1}^{+}P_{j+2}^{-} -P_{j}^{+}P_{j+1}^{-}P_{j+2}^{-}C^{\dagger}_{j,j+1,j+2} - C_{j,j+1,j+2}P_{j}^{+}P_{j+1}^{-}P_{j+2}^{-}
\end{align*}

\begin{align*}
H_X = \sum_{j=1}^{N-1}  P_{j}^{+}P_{j+1}^{-}\left[ \sum_a (T^a_{j} + T^a_{j+1})^2 \right] \\
H_\partial = P^{-}_{1} + P^{+}_{N} 
\end{align*}

All three of these operators manifestly commute with the $SU(k)$ generators $T^a = \sum_{j=1}^N T^a_j$, and with the operator $Z = \sum_{j=1}^N \sigma^z_j$, which together generate the symmetry group $SU(k) \times U(1)$ of the model, which can naturally be viewed as a full $U(k)$ symmetry. The Hilbert space at each site is then most naturally viewed as a sum of two mutually conjugate fundamental representations of U(k).

To illustrate the Hamiltonian more clearly, we can expand it in the product state basis as a sum of projectors. In that case, we can naturally write:

\begin{align*}
H_F = \sum_{j=1}^{N-2} \sum_{c_1,c_2,c_3} (\ket{\downarrow^{c_1}_j \uparrow^{c_2}_{j+1} \downarrow^{c_3}_{j+2}} -\ket{ \uparrow^{c_2}_{j} \downarrow^{c_3}_{j+1} \downarrow^{c_1}_{j+2}} )(\bra{\downarrow^{c_1}_j \uparrow^{c_2}_{j+1} \downarrow^{c_3}_{j+2}} -\bra{ \uparrow^{c_2}_{j} \downarrow^{c_3}_{j+1} \downarrow^{c_1}_{j+2}} ) + \\
(\ket{\uparrow^{c_1}_j \uparrow^{c_2}_{j+1} \downarrow^{c_3}_{j+2}} -\ket{  \uparrow^{c_2}_{j} \downarrow^{c_3}_{j+1}\uparrow^{c_1}_{j+2} } )(\bra{\uparrow^{c_1}_j \uparrow^{c_2}_{j+1} \downarrow^{c_3}_{j+2}} -\bra{  \uparrow^{c_2}_{j} \downarrow^{c_3}_{j+1}\uparrow^{c_1}_{j+2} } )
\end{align*}

\begin{align*}
H_X = \sum_{j=1}^{N-1} \left[ 1 - \frac{1}{k}(\sum_{c_1} \ket{\uparrow^{c_1}_j \downarrow^{c_1}_{j+1}}) (\sum_{c_2} \bra{\uparrow^{c_2}_j \downarrow^{c_2}_{j+1}}) \right] P_{j}^{+}P_{j+1}^{-}\\
H_\partial = P^{-}_{1} + P^{+}_{N} = \sum_{c_1} \ket{\downarrow^{c_1}_1}\bra{\downarrow^{c_1}_1} + \ket{\uparrow^{c_1}_N}\bra{\uparrow^{c_1}_N}
\end{align*}

where $c_1,c_2, c_3$ are colour indices running from 1 to k.\\

This model can be shown to reduce to the Fredkin model when $k = 1$, and in general the ground states can be obtained in a fully analogous manner and described in terms of combinatorial objects. The extension we will need is colored Dyck paths.

\subsection{Invariants of the Hamiltonian and characterization of the bulk ground states}

The basic idea here is the same as before: identifying basis states with paths, where we generalize the k=1 case by considering colored paths.\\

Our first move should then be to identify invariants of $H_F$ that may allow us to block diagonalize it. In the uncolored case, the only invariants were the number of unmatched up/down steps in the path notation or unmatched left/right parentheses in the word notation.\\

In the colored case, the Fredkin moves have more invariants, as we have already seen from the fact that the Hamiltonian commutes with the $T^a$ operators.\\

The key observation here is that the Hamiltonian generates the colored Fredkin moves, which we have presented in path notation in figure \ref{fig:ColoredMoves}. We see that these must preserve the color and order of the unmatched steps, as well as the color of the matched pairs.\\

Since we can move any $/ \backslash $ peak to any point in the path, we may use our previous algorithm where we repeatedly move the last peak to the beginning of the path until we are left with a series of single peaks followed by a dip, which up to permutations of the peaks will give the same result for any member of an equivalence class. This allows us to clearly see what equivalence classes we will be left with.\\

Specifically, the equivalence classes are fully specified by an ordered list of the colors of the unmatched up and down steps, along with the number of matched pairs corresponding to each color combination along the path.\\

To keep our approach SU(k)-covariant, it is convenient to work with basis states that are not quite product states when trying to use the orbit theorem. Since the colors of up and down steps form mutually conjugate fundamental representations of SU(k), they can combine into an SU(k) singlet or a member of the adjoint representation of SU(k). Our choice of basis states will thus correspond to paths($\sigma^z_j$ eigenstates), where the colors of each matched pair combine to either a color singlet $\ket{0}$, or to an element $\ket{a}$, $a \in [1,k^2-1]$ in the adjoint representation of SU(k).\\

Therefore, we may define the equivalence classes $C_{a,c_a,b,c_b,N^0,N^1 ... N^{k^2 - 1} }$, where a,b are analogous to in the uncolored case, $c_a, c_b$ are lists of length a,b, of the colors of the unmatched steps, $N^0$ is the number of matched pairs whose colors combine to singlets, and $N^a$ is the number of matched pairs that combine into the adjoint representation state a, and where we have the additional relation $2(N^0 + \sum_{a=1}^{k^2-1} N^a) + a + b = N$.\\

\begin{figure}
\includegraphics[width=0.65\textwidth]{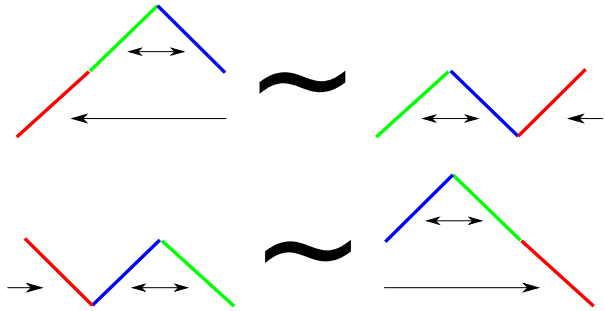}
\caption{\label{fig:ColoredMoves}The colored Fredkin moves in path notation. Note that the color of the matched pair and of the unmatched step is preseved, as well as the height of the unmatched step.}
\end{figure}

\subsubsection{The bulk ground states}

\begin{figure}
\includegraphics[width=0.65\textwidth]{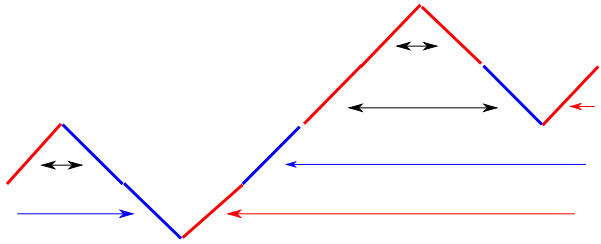}
\caption{\label{fig:ColoredPath}An example of a path corresponding to a product state basis element.}
\end{figure}

\begin{figure}
\includegraphics[width=0.65\textwidth]{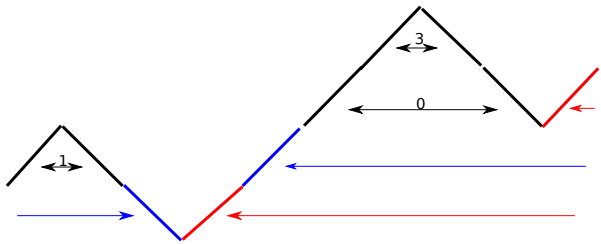}
\caption{\label{fig:CovariantPath}An example of a path corresponding to a covariant basis element, with one SU(k) singlet pair and two pairs in the SU(k) adjoint representation.}
\end{figure}

The colored Fredkin moves generates a group that permute our chosen SU(k)-covariant basis vectors, preserve the equivalence classes, and can map any element of an equivalence class to any other. Therefore, as a consequence of the orbit theorem each equivalence class will have a unique ground state $\ket{C_{a,c_a,b,c_b,N^0,N^1 ... N^{N^2 - 1} }}$ which is just the symmetric sum of all basis states in the class. These provide a nice orthogonal basis for all ground states of $H_F$.\\

This characterization is exactly what we will need to find out the combined ground state of $H_F + H_X + H_{\partial}$. The condition that the combined ground state is also a ground state of $H_X$ leaves us with $N^a = 0$ for all a, and clearly the entire subspace that does not fulfill this can not have any common $H_F + H_X$ zero eigenstate. Since in that case $N^0 = N - (a + b)$, we can simply call the bulk $H_F + H_X$ ground states $\ket{C_{a,c_a,b,c_b}}$. \\

For an even number of sites $N = 2n$ this leaves us with the unique ground state $\ket{C_{0,\emptyset,0,\emptyset}}$ as we add in the boundary terms, which we call the colored Dyck state $\ket{D^k_n}$.

\subsection{The colored Dyck state}

We have now fully introduced the ground state $\ket{D^k_n}$. Depending on whether we use the covariant or the product state basis. In the nonlocal covariant basis that we used in the previous section to simultaneously block-diagonalize $H_F$ and $H_X$, $\ket{D^k_n}$  is the symmetric sum of all paths colored such that all matched pairs form SU(k) singlets. \\

We may expand this expression as a sum of product states. In the local product state basis, $\ket{D^k_n}$ is instead the sum of all basis states corresponding to \textit{properly colored paths}, i.e. colored paths such that matching parentheses have the same color. The number of basis states in the sum is given by $k^n Cat(n)$. For k = 2, we will give some low-n examples in bracket notation:

\begin{align*}
\ket{D^2_2} = \frac{1}{\sqrt{2}} \left( \ket{()} + \ket{[]}  \right)\\ 
\ket{D^2_2} =  \frac{1}{\sqrt{8}} (  \ket{()()} + \ket{[]()} + \ket{()[]} + \ket{[][]}  +\\
\ket{(())} + \ket{([])} + \ket{[()]} + \ket{[[]]} ) 
\end{align*}

\begin{figure}
\includegraphics[width=0.65\textwidth]{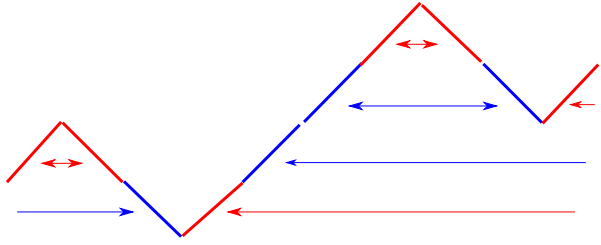}
\caption{\label{fig:pathOfColor}An example of a properly colored path.}
\end{figure}

\subsubsection{The Schmidt decomposition and measures of entanglement}

Just as before, our primary tool to analyze the entanglement of the colored Dyck state will be the Schmidt decomposition, which is given by

\begin{equation}
\ket{D^k_n} = \sum_{m = 0 }^{ L} \sum_{c \in \mathbb{Z}_k^m} \sqrt{p_{m,c}} \ket{ C_{0,\emptyset,m,c}(L) } \otimes \ket{ C_{m,c,0,\emptyset}(N-L) }
\end{equation}

where c is a list of length m consisting of the colors of the unmatched parentheses of $\ket{ C_{0,\emptyset,m,c}(L) }$ and $\ket{ C_{m,c,0,\emptyset}(N-L) }$. Since as before $$ p_{m,c}=\frac{|C_{0,\emptyset,m,c}(L)||C_{m,c,0,\emptyset}(N-L) |}{|C_{0,0,0,0}(N)|}$$, and since $\ket{D^k_n}$ is an equally weighted sum of every coloring of each path, we get $p_{m,c}$ is independent of c. By inserting the correct powers of k to account for every possible color combination we obtain $$p_{m,c} = p_m \frac{k^{\frac{L-m}{2}} k^{\frac{N-L-m}{2}}}{k^{\frac{N}{2}}} = k^{-m} p_m$$ where $p_m$ are the Schmidt coefficients from the uncolored case.\\

This allows us to compute the Schmidt rank and the Entanglement entropy in the colored case. We will start with the Schmidt rank. The number of nonzero $p_{m,c}$ coefficient is a geometric sum. In the case where L is even and $L = 2l$, the Schmidt rank is $\sum_{h=0}^{l} k^{2h} = \frac{k^{L+2} - 1}{k^2-1}$. In the case where L is odd and equal to $2l + 1$, we have $\sum_{h=0}^{l} k^{2h+1} = k \frac{k^{L+1} - 1}{k^2-1}$. The two may be combined into:

\[
\xi = k^{L\pmod 2} \frac{k^{ \myfloor{L+2} }-1}{k^2-1}
\]

Now we may move on to the entanglement entropy. By definition, we have:

\[
S = \sum_{m,c} -p_{m,c} \log(p_{m,c}) = \sum_m -p_m \log(k^{-m} p_m) = \sum_m m p_m \log(k) - p_m \log(p_m)
\]

which gives us

\begin{equation}
S = S_{1/2} + \log(k) \sum_{m = 0}^{L}  p_m m
\end{equation}

where the first term is just the log-scaling entanglement entropy from the uncolored case. The second term however, is the expectation value of the height at L for a random Dyck walk, which scales as $O(\sqrt[]{L})$ and is thus the key term leading to our violation of the area law.

\subsubsection{Computing the coefficient of the leading square root term in the entanglement entropy}

We are interested in computing the expression $\sum_{m = 0}^{L}  p_m m$. We proceed much like in the uncolored case. We consider the even case $L = 2l$ with nonzero $p_m$ when $m = 2h, h \in [0,l]$, and consider the coefficient $p_h \approx \frac{h^2}{Z} \exp\left[-h^2 \frac{n}{l(n-l)} \right]$, so that we are left with $2 \sum_h h p_h $. Replacing the sum with an integral gives us:

\begin{align*}
\sum_h h p_h = \int_{0}^{\infty} x p_x dx = \left[ \sqrt{ \frac{l(n-l)}{l}} d\alpha = dx  \right] = \int_0^{\infty}  \sigma \rho_\alpha \sigma^{-1} \alpha  \sigma^{-1} d\alpha \\
= \sigma^{-1} \int_0^{\infty} \alpha \rho_{\alpha} d \alpha = \sqrt{\frac{l(n-l)}{n}}  \int_0^{\infty} \frac{\alpha^3}{Z} \exp \left[ -\alpha^2 \right] d \alpha
\end{align*}
where just like before, $ p_x = \frac{x^2}{Z'}\exp\left[-x^2 \frac{n}{l(n-l)} \right]$, $\int p_x dx = 1$ ,  $\alpha = \sigma x $ and $\sigma = \sqrt{\frac{N}{L(N-L)} } $. We move to $\rho_{\alpha} = \frac{1}{Z'} \alpha^2 \exp[-\alpha^2]$ where $Z$ is a constant such that $\int \rho_{\alpha} = 1$, which gives us $ p_x = \sigma \rho_{\alpha}$ (with a factor $\sigma^3$ from the normalization factor and a $\sigma^{-2}$ from $ x^2 = \sigma^{-2} \alpha^2 $ ), which we used to insert the correct sigma factors. \\

The integral $\int_0^{\infty} \frac{\alpha^3}{Z} \exp \left[ -\alpha^2 \right] d \alpha$ is the third absolute central moment of the normal distribution, while in this formalism the constant $\frac{1}{Z}$ is chosen such that the second moment is one. Thus, we have:

\[
S \approx S_{1/2} + \log(k) \frac{M_3}{M_2}  \sqrt{2 \frac{L(N-L)}{L}} 
\]

where $M_2$ and $M_3$ are the second and third absolute moments of the normal distribution. Computing these gives us our final approximate expression:

\[
S \approx S_{1/2} +  \frac{2}{\sqrt{\pi}} \log(k) \sqrt{2 \frac{L(N-L)}{L}} 
\]

\bibliography{references}

\appendix
\section{The orbit theorem}

In this appendix, we will prove a useful theorem that we (implicitly or explicitly) repeatedly make use of throughout this paper. It is our main tool for classifying ground states by looking at equivalence classes. We chose to call it the orbit theorem. It reads as follows:

\begin{theorem}
Let $V$ be a Hilbert space of finite dimension, and let $v_1 ... v_n$ be an orthonormal basis of V. Let G be a finite group which acts on $B = \left\lbrace v_1, v_2 ... v_n \right\rbrace $ by permutations. Suppose that G is generated by the generators $g_1, g_2 ... g_n$ and let $P_1 , P_2 ... P_n$ be the corresponding permutation matrices. Let H be a Hamiltonian defined by $H = \sum_{i = 1}^m (1-P_i^{\dagger})(1-P_i) = \sum_{i = 1}^m 2 -P_i^{\dagger} - P_i$ and let $W \subset V$ be the zero eigenspace of H. \\

Then H is positive semidefinite, is a frustration free Hamiltonian, and the dimension of the ground state space W is equal to the number of orbits under the action of G on B. Furthermore, The Hamiltonian can be block diagonalized into the different vector spaces $V_i$ spanned by the elements of $B_i$, and each block has a unique ground state which is a zero eigenstate of H. The ground states of the different blocks form an orthogonal basis for W.
\end{theorem}
\begin{proof}
H is clearly positive semidefinite since it is a sum of the form $\sum_i A_i^{\dagger} A_i$. Let $B_1, B_2 ... B_k$ be the orbits of G and $V_i$ the spaces spanned by the elements of $B_i$. Then the linear representation of G on V decomposes into representations of G on the subspaces $V_1 ... V_k$ such that W is a direct sum of the spaces $W_i := W \cap V_i$. and $dim W = dim W_1 + dim W_2 + ... + dim W_k$. It is thus enough to show that $dim W_i = 1$ for all i, or after change of notation that dim W = 1 when G acts transitively on V (if G can map any element to any element). Suppose that $\ket{v} \in W$. Then we have $\sum_i [\bra{v} (1-P_i^{\dagger}) ] [(1-P_i) \ket{v} ] = \bra{v} H \ket{v} = 0 $  which implies that $(1 - P_i) \ket{v} = 0$, which finally gives us $P_i \ket{v} = \ket{v}$ for all i. Any $v \in W$ will therefore be fixed by all permutations generated by the $P_i$ matrices, and since they can map any basis state to any other basis state (act transitively) the coordinates of v corresponding to any two basis states must coincide. So dim W = 1 as the equally weighted sum of all basis states is the only ground state of H in V, as desired. The last part of the theorem then simply follows by noticing that the sets $B_i$ of basis vectors do not intersect, therefore their spanned vector spaces $V_i$ are orthogonal to each other.
\end{proof}

We note that in all cases where we apply this theorem, the permutation matrices (corresponding to Fredkin moves) happen to be involutions(and as such have eigenvalues $\pm 1$, so $ (1-P_i^{\dagger})(1-P_i) = 2(1-P_i)$, which we for convenience normalize as $1-P_i$.

\end{document}